%% file: himeasurement.tex
\documentclass[conference,a4paper]{IEEEtran}

\usepackage[utf8]{inputenc} 
\usepackage[T1]{fontenc}
\usepackage{ifthen}
\usepackage{cite}
\usepackage[cmex10]{amsmath} 

\usepackage[english]{babel} 
\usepackage{uniinput} 

\usepackage[table,usenames,dvipsnames]{xcolor}
\usepackage{cite} 

\usepackage{booktabs}

\usepackage{tikz}
\usetikzlibrary{backgrounds}
\definecolor{lightblue}{rgb}{0.3,0.7,0.95}

\usepackage{amssymb} 
\usepackage{mathtools} 
\usepackage{amsmath} 
\usepackage{amsfonts}
\usepackage{amsthm} 

\usepackage{algorithmic} 
\usepackage{algorithm} 
\usepackage{balance}
\usepackage{complexity}

\definecolor{ingo}{RGB}{139,0,0}

\newcommand{\B}{{\bf B}}
  
\usepackage[normalem]{ulem} 
\usepackage{cancel} 
\usepackage{complexity}
\input{mymath}

\begin{document}
\title{Hierarchical restricted isometry property for Kronecker product measurements} 



 \author{%
   \IEEEauthorblockN{Ingo Roth\IEEEauthorrefmark{1},
                     Axel Flinth\IEEEauthorrefmark{2},
                     Richard Kueng\IEEEauthorrefmark{3},
                     Jens Eisert\IEEEauthorrefmark{1}, 
                     and Gerhard Wunder\IEEEauthorrefmark{1}}
   \IEEEauthorblockA{\IEEEauthorrefmark{1}%
                     Freie Universit\"{a}t Berlin, 
                     \{i.roth, jense, g.wunder\}@fu-berlin.de}
   \IEEEauthorblockA{\IEEEauthorrefmark{2}%
                     Technische Universit\"{a}t Berlin, 
                     flinth@math.tu-berlin.de}
   \IEEEauthorblockA{\IEEEauthorrefmark{3}%
                     California Institute of Technology, Pasadena, 
                     rkueng@caltech.edu}
 }

\maketitle

\begin{abstract}
%
Hierarchically sparse signals and Kronecker product structured measurements arise naturally in a variety of applications. The simplest example of a hierarchical sparsity structure is two-level $(s,\sigma)$-hierarchical sparsity which features $s$-block-sparse signals with $\sigma$-sparse blocks. For a large class of algorithms recovery guarantees can be derived based on the restricted isometry property (RIP) of the measurement matrix and model-based variants thereof. 
We show that given two matrices $\vec A$ and $\vec B$ having the standard $s$-sparse and $\sigma$-sparse RIP their Kronecker product $\vec A \otimes \vec B$ 
has two-level $(s,\sigma)$-hierarchically sparse RIP (HiRIP).
This result can be recursively generalized to signals with multiple hierarchical sparsity levels and measurements with multiple Kronecker product factors. As a corollary we establish the efficient reconstruction of hierarchical sparse signals from Kronecker product measurements using the HiHTP algorithm. 
We argue that Kronecker product measurement matrices allow to design large practical compressed sensing systems that are deterministically certified to reliably recover 
signals in a stable fashion. We elaborate on their motivation from 
the perspective of applications.
\end{abstract}



\section{Introduction}

The field of compressed sensing studies the solution of the underdetermined inverse problem of reconstructing a suitably structured signal $\x \in \KK^d$ from linear noisy samples $\y = \A\x + \vec{e} \in \KK^{m}$, where $\A \in \KK^{m\times d}$, $m < d$, is a measurement matrix and $\vec{e}$ accounts for additive noise.  The most prominent structure assumption on $\x$ is thereby sparsity. By $\KK$ we denote a field that is either that of 
real numbers $\RR$ or of complex numbers $\CC$. 

The recovery of $\x$ from $\y$ and $\A$ is guaranteed with high probability for a variety of algorithms when the measurement matrix $\A$ is drawn from a suitable random ensemble. A working-horse in proving such recovery guarantees is that a random measurement matrix $\A$ often fulfills the so-called
\emph{restricted isometry property} (RIP) with high probability. This means that there exist $\delta \in [0,1)$ such that 
\begin{equation}
  (1-δ) \norm{\x}^2 \leq \norm{\A\x}^2 \leq (1+δ) \norm{\x}^2
\end{equation}
for all $s$-sparse $\x \in \KK^d$. Here $\norm{\x}^2 = \sum_{i=1}^d |x_i|^2$ denotes the $\ell_2$-norm. Typical examples of such measurement ensemble fulfilling a RIP ensuring reconstruction from $m \geq m_0$ samples with $m_0 \in \tilde{\mathcal{O}}(s)$ are sub-Gaussian matrices or subsampled Fourier matrices. In practice, however, 
 matrices are most often `less random'. 
 In fact, in many applications it is highly desirable to make use of as little randomness as possible.
This work will focus on measurement matrices that can be written as the Kronecker product of smaller matrices. 
 In this sense, this work contributes to the broader scheme of partially derandomising 
recovery schemes. 
\subsection{Kronecker product measurements}

Measurement matrices that are the Kronecker product of a number of smaller matrices naturally appear in 
various practical applications. As an illustrative example let us consider the following simple multi-user communication model: A potentially very large number $N$ of users simultaneously send messages $\vec x_i$ of length $n$ to a central base station. They thereby encode their messages with a common compressed sensing matrix $\vec A$. At the base station $m$ different superpositions $y_j = \sum_{i} b_{j,i} \vec A \vec x_i$ of the individual encoded messages $\vec A \vec x_i$ are measured. One can, for instance, think of a massive MIMO system, where the different weights $b_{j,i}$ arise from the fact that the encoded messages $\vec A\vec x_i$ scatter along different paths to arrive at the base station. Hence, at the base station we want to recover the entire signal $\vec x=[\vec x_1^T, \ldots, {\vec x_N}^T]^T$ from the linear measurements of the form $\vec y = (
\vec B \otimes \vec A) \vec x$ where $\vec B$ is the $m \times N$ matrix with entries $b_{j,i}$. We conclude that Kronecker product measurements are typically encountered when the superposition of multiple parties that share a common sensing/coding matrix are observed.

Another important class of examples is constituted by unit rank measurements on matrices as they can be cast as Kronecker product measurements. Consider measurements on $\vec X \in \KK^{N, n}$ of the form $Y_{i,j} = \Tr(\vec a_i \vec b_j^T \vec X)$, where $\vec a_i$ and $\vec b_j$ denote the columns of a matrix $\vec A \in \KK^{M\times N}$ and $\vec B \in \KK^{m\times n}$, respectively. Then using column-wise vectorisation it holds that $\vecmap(\vec Y) = \vec A \otimes \vec B^T \vecmap(\vec X)$. 
Such unit rank measurements often arise in bilinear compressed sensing problems that are lifted 
\cite{PhaseLift} to linear matrix problems. The results of this work are for example applied in angle-delay pair estimation in massive MIMO in Ref.~\cite{mmimoHiHTP} along those lines.





From the computational perspective, Kronecker products have a number of highly desirable 
properties. For instance, they can to some extent be applied in parallel computations or stored more efficiently. A Kronecker product $\A \otimes \B \in \KK^{M\times N} \otimes \KK^{m\times n}$ is described by $MN + mn$ parameters, whereas a general matrix $\vec C \in \KK^{Mm\times Nn}$ needs $MNmn$ parameters. At the same time, this significantly
reduces the amount of randomness that is required to generate such matrices. This is, in fact, an obstacle for proving that such matrices obey the standard RIP property.

The relation between the RIP-constants of a group of matrices $\vec A_1, \dots, \vec A_L$ and the corresponding constant for the Kronecker product $\vec A_1 \otimes \dots \otimes \vec A_L$ have been investigated in Refs.\  \cite{JokarMehrmann,KroneckerCS}. Therein, the authors use a slightly different convention for the RIP constants, namely
\begin{equation}
  (1-\overline{\delta})\norm{\x} \leq \norm{\A\x} \leq (1+\overline{\delta})\norm{\x}. \label{eq:RIPalt}
\end{equation}
In short, a Kronecker product has the $k$-RIP if and only if each of its blocks has the $k$-RIP. More concretely, 
\begin{align*}
  \max_{1 \leq l \leq L} \overline{\delta}_k(\A_\ell) \leq \delta_k( \A_1 \otimes \dots \otimes \A_L ) \leq \prod_{\ell=1}^L (1+ \overline{\delta}_k(\A_\ell)) -1.
\end{align*}
For us, in particular the lower bound is interesting. It tells us that
\emph{if we intend to build a matrix $\A_1 \otimes \dots \otimes \A_L$ with the $s$-RIP, we need each $\A_\ell$ to 
exhibit the $s$-RIP!} We can in particular not ensure to be able to reconstruct arbitrary signals of higher sparsity than $s$ if  one of the matrices $\A_\ell$ fails to reconstruct $s$-sparse signals.

\subsection{Hierarchically sparse vectors}
Motivated by a variety of applications, more restricted sparsity structures have intensively been
studied over the last decade.
Classic examples of structured sparse signals are signals that have only a small number of non-vanishing but possibly dense blocks, block-sparsity, \cite{EldarMishali2009,StojnicParvareshHassibi2009} or signals that feature sparse blocks (see, e.g., Ref.~\cite{AdcockEtAl2013}). The combination and generalisation of these structure leads to the concept of hierarchically sparse signals. 
The simplest example are \emph{two-level $(s,\sigma)$-hierarchically sparse} signals (see, e.g., Refs.~\cite{
SprechmannEtAl2010,FriedmanEtAl2010,HiRIP}).

More precisely, let $\x \in \KK^{Nn}$. We can partition $\x$ into $N$ blocks $\x_i$, each of size $n$. 
\begin{definition}[Hierarchical sparsity]
  A  vector $\x \in\mathbb{K}^{nN}$ is \emph{$(s,\sigma)$-hierarchically sparse} if at most $s$ blocks have non-vanishing entries and each of these blocks is $\sigma$-sparse.
\end{definition}
For convenience, we will call a hierarchically $(s,\sigma)$-sparse vector
simply \emph{$(s,\sigma)$-sparse} in this work. 
In the applications discussed above, hierarchically sparse signals are a reasonable restriction. In our simple communication model for example, an $(s,\sigma)$-sparse signal $\vec x$ arises if we demand that at a given time only a maximum of $s$ users are active and the messages $\vec x_i$ itself are each $\sigma$-sparse.
The exploitation of such finer structure assumptions has been identified as crucial in the development of future scalable mobile communication systems \cite{WunderEtAl2014,WunderEtAl2015} and they have been studied in the task of channel estimation and user activity detection, e.g., in Ref.~\cite{WunderEtAl2017}.

Similarly, in bilinear compressed sensing problems where both arguments are sparse, 
the resulting vectorisation of the lifted matrix is hierarchically sparse \cite{mmimoHiHTP,WunderEtAl2018}. Note that in lifted problems the signals will also have a low-rank structure and thus be more structured than being
merely hierarchically sparse. For this reason, hierarchically sparse recovery methods are not expected to achieve an information theoretically optimal sampling complexity in these settings. But they are still of interest because of their low computational demands.  These are important examples; it goes without saying that hierarchically
sparse signals are ubiquitous in signal processing, in physics and in the life sciences.

By adopting the notion of model-based sparse recovery \cite{ModelbasedCS}, three of the five authors of this paper designed an iterative thresholding algorithm, HiHTP, for recovering $(s,\sigma)$-sparse and more general hierarchically sparse vectors, see
Ref.\  \cite{HiRIP}. The algorithm follows the same strategy as the original hard-thresholding pursuit (HTP) algorithm of \cite{Foucart:2011}.  In every iteration it estimates the support using a thresholding operation on a gradient step and subsequently solves the least-squares fitting problem restricted to the estimated support.  The main modification for the recovery of $(s,\sigma)$-sparse vectors is to employ the projection onto vectors with $(s,\sigma)$-sparse support 
\begin{equation}
  L_{s,\sigma}(\x) \coloneqq \supp \argmin_{\text{$(s,\sigma)$-sparse $\z$}} \norm{\x -\z}. 
\end{equation}
As argued in Ref.\ \cite{HiRIP} this projection can be efficiently calculated. Algorithm~\ref{alg:HiHTP} shows the resulting HiHTP algorithm.
\begin{algorithm}[tb]      
  \caption{(\HiHTP)} 
  \label{alg:HiHTP}
  \begin{algorithmic} [1]
    \REQUIRE measurement matrix $\A$, measurement vector $\y$, block column sparsity $(s,\sigma)$
    \STATE $\x^0 = 0$ 
    \REPEAT
      \STATE $\Omega^{k+1} = L_{s,\sigma} (\x^k + \A^\ast (\y - \A \x^k))$\label{alg:HiHTP:TH}
      \STATE $\x^{k+1} = \argmin_{\vec z \in \CC^{Nn}} \{ \|\y - \A\vec z\|,\  \supp(\vec z) \subset \Omega^{k+1} \}$ \label{alg:HiHTP:LS}
    \UNTIL stopping criterion is met at $\tilde{k} = k$
    \ENSURE $(s,\sigma)$-sparse vector $\x^{\tilde{k}}$
  \end{algorithmic}
\end{algorithm}
 The algorithm was proven to converge to the correct signal under an \emph{HiRIP}-assumption on the measurement matrix. A matrix is thereby said to have the HiRIP property if an inequality like \eqref{eq:HiRIP} is satisfied for all $(s, \sigma)$-sparse $\x$. Since the set of $(s,\sigma)$-sparse vectors is contained in the set of $s\cdot \sigma$-sparse vectors, the HiRIP is a weaker condition compared to standard RIP.

This work is dedicated to deriving
statements about the HiRIP-properties of Kronecker products $\A \otimes \B$. We will prove that the $(s,\sigma)$-HiRIP constant of $\A \otimes \B$ is bounded by the $s$-RIP constant of $A$ and $\sigma$-RIP constant of $B$ as follows:
\begin{align*}
\delta_{(s,\sigma)}^{\A \otimes \B} \leq \delta_{s}^{\A} + \delta_{\sigma}^{\B} + \delta_s^{\A} \delta_\sigma^{\B}.
\end{align*}
Hence, the Kronecker product of matrices with good RIP-constants has a non-trivial HiRIP constant.
This is in sharp contrast to the properties of the RIP discussed above. This discrepancy indicates that one can derive much stronger recovery results when dealing with hierarchical sparsity patterns, rather than unstructured ones.

We will argue that a similar statement holds for multilevel hierarchical structures and Kronecker products of the form $\vec A_1 \otimes \dots \otimes \vec A_L$. So we find that that tensor products 
in general inherit multilevel HiRIP from the RIP of the constituents. 

From a more information theoretic perspective, the result opens up a new possibility to actually certify HiRIP for a given matrix. In principle, given a matrix $\vec A$, a sparsity level $s$ and a constant $\delta > 0$, it is an NP-hard problem to decide whether the RIP constant $\delta_s$ of $\vec A$ is smaller than $\delta$ \cite{TillmannPfetsch2014, BandeiraEtAl2013}. 
However, our result indicates that certifying that a HiRIP constant of the matrix $\vec A$ is smaller than
$\delta$ can be done by checking that all $\vec A_i$ have (sufficiently) smaller RIP-constants. If the dimension of the matrices $\vec A_i$ are small enough it is even practical to certify their RIP by brute-force calculations of a spectral norm for all possible sparse supports. In fact, we find that in certain parameter regimes the complexity of certification of HiRIP of a matrix $\vec A=\vec A_1^{\otimes l}$ with the brute-force algorithm scales polynomial in the size of the large matrix $\vec A$.

The rest of the paper is organized as follows.  In Section \ref{sec:main}, we present the technical statement of our main results and some fundamental consequences of them. In Section \ref{sec:check}, we discuss on how our main result can be used to design matrices that are known to have HiRIP. 

\section{Main Results} \label{sec:main}
In the following, let $[d]$ be the subset $\left\{1, \ldots, d\right\} \subset \NN$ of integers smaller or equal than $d \in \NN$. Furthemore, for $\x \in \KK^{d}$ and $\Omega \in [d]$  we define   the vector $\x|_\Omega$ that coincides with $\x$ on the indices in $\Omega$ and vanishes otherwise. 
Let us begin by formally defining the HiRIP for $(s,\sigma)$-hierarchical sparsity.

\begin{definition}[HiRIP]
  Given a matrix $\A \in \KK^{m\times nN}$, we denote by $\delta_{s,\sigma}$ the smallest $\delta \geq 0$ such that 
  \begin{equation}
    (1-\delta)\|\x\|^2 \leq \|\A\x\|^2 \leq (1+\delta) \|\x\|^2 \label{eq:HiRIP}
  \end{equation}
  for all $(s,\sigma)$-hierarchically sparse vectors $\x \in \KK^{nN}$.
\end{definition}


As has been advertised in the introduction, we can prove the following result

\begin{theorem}[Main result]\label{thm:hirip}
  Given $\A \in \KK^{M\times N}$ having $s$-sparse RIP with constant $δ^\A_s$ and $\B \in \KK^{m\times n}$ with $σ$-sparse RIP with constant $δ^\B_\sigma$, then 
  \begin{equation}
    \A \otimes \B : \KK^{Nn} \to \KK^{Mm}
  \end{equation}
  has $(s,\sigma)$-sparse HiRIP with constant 
  \begin{equation}
     δ_{(s,σ)} \leq δ^\A_s + δ^\B_σ + δ^\A_sδ^\B_σ.
   \end{equation} 
\end{theorem}


  Before presenting the proof of Theorem \ref{thm:hirip}, we need to introduce some 
  notation. First, we let $\vecmap: \KK^{N\times n} \to \KK^{Nn}$ denote the canonical isomorphism of column-wise vectorisation. In other words, $\vecmap$ is defined by linear extension of the requirement $\vecmap(\E_{i,j}) = \vec{e}_i \otimes \vec{e}_j$, where $\E_{i,j} = \vec{e}_i \vec{e}_j^T \in \KK^{N\times n}$ denotes the matrix with only one non-vanishing unit entry in the $i$-th row and $j$-th column. The Kronecker product is always understood as 
\begin{equation}
  \A \otimes \B = \begin{pmatrix} 
    a_{1,1} \B& \ldots & a_{1,N}\B  \\
    \vdots & \ddots & \vdots \\
    a_{m,1} \B & \ldots & a_{m,N}\B
  \end{pmatrix}.
\end{equation}
This convention justifies the term \emph{column-wise vectorisation}. 

It will be convenient to also implicitly make use of \emph{row-wise vectorisation}, which can be defined as $\vec X \mapsto \vecmap(\vec X^T)$.  Passing from one vector representation to the other amounts to applying the \emph{flip operator} $F_{N,n}: \KK^{Nn} \to \KK^{Nn}$, that linearly extends the mapping $\vec{e}_i \otimes \vec{e}_j \mapsto \vec{e}_j \otimes \vec{e}_i$.

The action induced by switching between the column-wise to row-wise vectorisation in the space of operators acting on the vector space is the swap of the tensor product components. To be precise:
\begin{lemma}\label{lem:flipping} For $\A \in \KK^{M\times N}$ and $\B \in \KK^{m\times n}$ and $\vec X \in \KK^{N\times n}$ it holds that 
\begin{equation*}
  (\A \otimes \B) \vecmap(\vec X) = F_{m,M} (\B \otimes \A) F_{N,n} \vecmap(\vec X)
\end{equation*}
and
\begin{equation*}\label{eq:fliptransp}
  F_{N,n} \vecmap(\vec X) = \vecmap(\vec X^T). 
\end{equation*}
\end{lemma}

We now have all the tools we need to prove Theorem \ref{thm:hirip}.

%
%
%
%

\begin{proof}[Proof of Theorem~\ref{thm:hirip}]
Let $\x \in \KK^{Nn}$ be hierarchically $(s,\sigma)$-sparse. With the help of Lemma~\ref{lem:flipping} we find
\begin{align*}
  \norm{(\A \otimes \B)\x}^2  &= \norm{(\A \otimes \Id_n)(\Id_N \otimes \B) \x}^2 \\
  &=\norm{F_{m,M} (\Id_n \otimes \A) F_{N,n} (\Id_N \otimes \B) \x}^2 \\
  &= \norm{(\Id_n \otimes \A) F_{N,n} (\Id_N \otimes \B) \x}^2,
\end{align*}
where the last line follows from the fact that $F_{m,M}$ is unitary.  The vector $(\Id_N \otimes \B)\x$ has only non-vanishing entries in $s$ of its $N$ blocks. Therefore, the flipped vector $\vec{h} \coloneqq F_{N,n}(\Id_N \otimes \B)\x$ consists of blocks $\vec{h}_i \in \KK^N$ with $i \in [n]$ that are at most $s$-sparse each. This allows us to apply the $s$-sparse RIP property of $\A$ for each of the blocks
\begin{align*}
  &\norm{(\Id_n \otimes \A) \vec{h}}^2 = \sum_{i\in [n]} \norm{\A \vec{h}_i}^2  \leq (1+ \delta_s) \norm{\vec{h}}^2.
\end{align*}
Making use of the unitarity of the flip once again, the $\ell_2$-norm of $\vec{h}$ is identical to 
\begin{align} \label{eq:crucialStep}
  \norm{\vec{h}}^2 &= \norm{(\Id_N \otimes \B)\x}  = \sum_{i \in [N]} \norm{\B\x_i},
\end{align}
where $\x_i \in \KK^n$ $i \in [N]$ are the $\sigma$-sparse blocks of $\x$. Every term of the sum is bounded by the $\sigma$-sparse RIP of $\B$ yielding 
\begin{equation*}
  \norm{\vec{h}}^2 \leq (1+\delta_\sigma)\norm{\x}^2. 
\end{equation*}
In summary, we have established
\begin{equation*}
  \norm{(\A\otimes \B)\x}^2 \leq (1+\delta_s) (1+ \delta_\sigma) \norm{\x}.
\end{equation*}
The lower RIP bound can be derived in the same way, completing the proof. 
\end{proof}




The main consequence of Theorem~\ref{thm:hirip} is that it allows to construct a new class of measurement matrices for which the HiHTP algorithm is guaranteed succeed. More precisely, we get the following corollary. 

\begin{corollary}\label{cor:recovery}
  Let $\A \in \KK^{M\times N}$ and $\B \in \KK^{m\times n}$, and suppose that the following RIP-conditions hold
  \begin{align*}
    \delta^{\A}_{3s}, \delta^{\B}_{2\sigma} \leq \sqrt{\frac{\sqrt{3}+1}{\sqrt{3}}} -1. 
  \end{align*}
  Then, for $\x \in \KK^{nN}$, $\vec e \in \KK^{Mm}$ and $\Omega \subseteq [N] \times [n]$ and $(s,\sigma)$-sparse support set, the sequence $\x^k$ defined by the HiHTP Algorithm~\ref{alg:HiHTP} with $y = (\A \otimes \B) \x\vert_\Omega +\ev$ satisfies, for any $k\geq 0$
  \begin{align*}
    \norm{ \x^k - \x \vert_\Omega} \leq \rho^k \norm{\x^0- \x \vert_\Omega} + \tau \norm{\ev},
\end{align*}   
where
\begin{align*}
  \rho = \left(\frac{2(\delta^{\A}_{3s}+ \delta^{\B}_{2\sigma} + \delta^{A}_{3s} \delta^{\B}_{2\sigma})} {1 -(\delta^{\A}_{3s}+ \delta^{\B}_{2\sigma} + \delta^{A}_{3s} \delta^{\B}_{2\sigma})^2 }\right) <1.
\end{align*}
\end{corollary}
\begin{proof}
  We simply need to note that Theorem \ref{thm:hirip} implies that $$\delta_{3s,2\sigma}^{\A \otimes \B} \leq (\delta^{\A}_{3s}+ \delta^{\B}_{2\sigma} + \delta^{A}_{3s} \delta^{\B}_{2\sigma} )< \frac{1}{\sqrt{3}}.$$
  The rest follows from Theorem~1 of \cite{HiRIP}.
\end{proof}

In Ref.~\cite{HiRIP} one example has been given of a class of matrices guaranteed to possess the HiRIP with high probability. Concretely, it was shown that a random dense Gaussian matrix $\vec G\in \KK^{m\times Nn}$ has the $(s,\sigma)$-HiRIP with high probability under the assumption
\begin{align*}
  m \gtrsim s\sigma \log(N) + \sigma \log(Nn).
\end{align*}
These are slightly less measurements than the $s\sigma \log(Nn)$ measurements needed to secure the (unstructured) $s\sigma$-RIP with high probability \cite{RIP}. 

We can now describe a new class: Taking any pair of random matrices $\A \in \KK^{M\times N}$ and $\B \in \KK^{m\times n}$ both guaranteed to possess the $s$- and $\sigma$-RIP with high probability, $\A \otimes \B$ will have the $(s,\sigma)$-HiRIP. As an example, we can use random Gaussian matrices with $M \gtrsim s \log(N)$ and $m \gtrsim \sigma \log(n)$, resulting in a measurement matrices $\vec A \otimes \vec B \in \KK^{\mu\times nN}$ with $$\mu \gtrsim s\sigma \log(N)\log(n).$$
Hence, a measurement scheme using Kronecker matrices will need slightly more measurements than the fully Gaussian matrices to have the HiRIP. 

This price could, however, sometimes be worth paying. First, we get a vast reduction in space needed to store the matrix ($(MN + mn)$ instead of $MN\cdot mn$). Also, as has been discussed in the introduction, there are applications where the Kronecker structure of a measurement process is inherent. 
Moreover, the results and notions of this section can readily be generalized to hierarchical sparsity with more than two levels. 

\begin{definition}[Multilevel hierarchical sparsity, HiRIP]
  Let $L \geq 3$, $n_1, \ldots, n_L$ and $s_1, \ldots, s_L$ be natural numbers.
 \begin{enumerate}
 \item A vector $\x \in \KK^{n_1  \cdots  n_L}$ is called $(s_1, \ldots, s_L)$-(hierarchically) sparse if it consists of $n_1$ blocks $\x_i \in \KK^{n_2  \cdots n_L}$ such that only $s_1$ blocks are non-zero, and each $\x_i$ is $(s_2, \dots, s_L)$-sparse.
 \item For $\A^{m \times n_1 \cdots n_L}$, we define $\delta_{s_1, \dots, s_L}$ as the smallest $\delta \geq 0$ for which Inequality~\eqref{eq:HiRIP} holds for all $(s_1, \dots, s_L)$-sparse vectors $\x$.
 \end{enumerate}
\end{definition}

The following result is a generalization of  Theorem~\ref{thm:hirip}.

\begin{theorem} \label{thm:HiRIPMoreLevels}
  Let $\vec A \in \mathbb{K}^{m\times n_1}$ be a matrix with  RIP-constant $\delta_{s_1}^A$ and $\vec B \in \KK^{M\times N}$ one with HiRIP constant $\delta_{s_2, \dots, s_L}^{\vec B}$. Then the hierarchical RIP-constant $\delta_{s_1, \dots, s_n}$ of $\vec A \otimes \vec B \in \KK^{mM\times nN}$ satisfies
  \begin{align*}
    \delta_{s_1, \dots , s_L}\leq  \delta_{s_1}^{\vec A} + \delta_{s_2, \dots, s_L}^{\vec B} + \delta_{s_1}^{\vec A}  \cdot \delta_{s_2, \dots, s_L}^{\vec B}.
  \end{align*}
  In particular, through induction, we obtain for matrices $\vec A_i \in \KK^{m_i\times n_i}$, $i=1, \dots, L$, with $s_i$-th RIP constants $\delta_{s_i}^{\vec A_i}$:
  \begin{align*}
    \delta_{s_1, \dots, s_L}(\vec A_1 \otimes \dots \otimes \vec A_L) \leq \prod_{i=1}^L (1+ \delta_{s_i}^{\vec A_i}) -1.
  \end{align*}
\end{theorem}

The techniques of the proof of Theorem \ref{thm:hirip} can readily be adapted to prove also Theorem \ref{thm:HiRIPMoreLevels}. 

\begin{proof}[Proof of Theorem~\ref{thm:HiRIPMoreLevels}]
The proofs reads exactly as the proof of Theorem~\ref{thm:hirip} up to equation \eqref{eq:crucialStep}, with an adapted version of the flipping operator. Here, we use that the blocks $\x_i$ are not $\sigma$-sparse, but $(s_2, \dots, s_L)$-sparse, and apply the corresponding RIP of $\B$. The statement now follows in exactly the same manner as above. 
\end{proof}

Analogously, to Corollary~\ref{cor:recovery} we can derive a recovery guarantee for a generalisation of HiHTP to more general hierarchically sparse vectors including multiple layers using the results of Ref.~\cite{HiRIP}.

\section{Efficiently checkable HiRIP} \label{sec:check}
Theorem~\ref{thm:HiRIPMoreLevels} has a consequence that may be surprising. Given a matrix $\vec A \in \KK^{M\times N}$ and some $\delta > 0$, to certify whether the standard RIP constant $\delta_s$ is smaller than $\delta$ is in general an NP-hard problem \cite{BandeiraEtAl2013, TillmannPfetsch2014}. 
To this date, no deterministic constructions of measurement matrices are known that achieve an optimal scaling in sampling complexity. 
For a variety of  ensembles of matrices there exist guarantees that with high probability a random instance with optimal scaling complexity fulfils the RIP. However, checking whether one was lucky or not is not feasible already for intermediate sized systems.  

The brute-force approach for certifying the normal $S$-sparse RIP of a matrix $\vec A \in \KK^{M\times N}$ is to calculate $\|{\vec A_\Omega\ad \vec A_\Omega - \Id_S}\|$ for all $\binom{N}{S}$ $S$-sized supports  $\Omega$ and taking the maximum. The computational effort of this approach scales at least as $(N/S)^S S^3$, which grows exponentially in $N$ for fixed ratios $(N/S)$. 
Note that the number of hierarchically sparse supports also scales exponentially in the overall system size. Hence, a brute force calculation to certify the HiRIP of an arbitrary matrix is also exponentially expensive.

If we, however, let $\vec A = \vec B^{\otimes L}$, where $\vec B \in \KK^{m\times n}$. We can certify the $\vec s=(s,\dots,s)=:(s^{,L})$-HiRIP of $\vec A$ by simply brute-force checking the $s$-RIP of the matrix $\vec B$, and subsequently invoking Theorem~\ref{thm:HiRIPMoreLevels}, which in this case reads: 
\begin{corollary}\label{cor:levelHiRIP}
  Given a matrix $\vec B$ with $s$-sparse RIP with constant $\delta_s$, the matrix $\vec B^{\otimes L}$ has $(s^{,L})$-HiRIP with 
  \begin{equation}
    \delta_{(s^{,L})} \leq (1+ \delta_s)^l - 1.
  \end{equation}
\end{corollary}

The brute-force calculation of the $s$-sparse RIP of $\vec B$ thereby only takes an order of $(n/s)^s s^3 =(N/S)^\frac{S^{1/L}}{L} S^{3/L}$ computations, where $N=n^L$ and $S=s^L$ denotes the total system size and total sparsity, respectively. In the regime where $S \leq c L^L$ for some constant $c \in \RR$, we arrive at a polynomial scaling in the overall system size $N$. Hence, we can certify that $\vec A = \vec B^{\otimes L} \in \KK^{M\times N}$ has $(s^{,L})$-HiRIP in an efficient way.  More generally, our approach applies to all measurement matrices $\vec A$ with a known tensor decomposition of the form $\vec M = \sum_{i=1}^r \vec M_{i_1} \otimes \ldots \otimes \vec M_{i_L}$. Checking RIP individually for all compound matrices $\vec M_{i_j}$ induces an additional factor  of $rL$ in the computational costs of certifying HiRIP for $\vec M$ compared to a matrix of the form $\vec B^{\otimes L}$. 

Let us verify that this efficient scheme is actually practical for real applications by evaluating the computational costs for reasonable parameter values. Let us assume $L=3$ levels with block
size $n=10^2$ and sparsity $s = 10$ on each level.This amounts to  $N = 10^6$ and $S = 10^3$. Hence, the brute-force approach for checking RIP or HiRIP for an arbitrary measurement matrix requires order of $10^{10^3}$ computations. In comparison, if we let  as many of today's fastest computing devices (with around 100 Peta FLOPS) as there are atoms in the universe run for the entire estimated age of the universe, we would be able perform around $10^{115}$ computations.
On the contrary, we can check $\delta_{10}$ for $\vec B$ with $n = 10^2$ with order of $10^{10}$ computations, which is practically feasible on current desktop hardware.

This scheme, however, has a catch in the form of a suboptimal sampling complexity $M$. To ensure that $δ_{(s^{,L})}$ is smaller than a constant $\tilde\delta$ using Corollary~\ref{cor:levelHiRIP} one has to have $δ_s \leq (\tilde{δ} +1)^{1/L}-1$. Typically, for a matrix $\vec B \in \KK^{m\times n}$ to have RIP with $δ_s \leq δ$ requires $m \geq m_0$ with $m_0$ scaling at least as $m_0 \sim 1 / \delta \sim L/\log(1+\tilde{\delta}) - 1/2 + \mathcal{O}(1/L)$ for large $L$. Therefore, in the regime of efficiently checkable HiRIP $S \leq c L^L$, we get an overall sampling complexity $M = m^L \gtrsim (Ls)^{L}  \gtrsim S^2$ scaling quadratically in $S$. 
This is reminiscent of the quadratic bottleneck that also plagues most deterministic constructions of RIP matrices.


%
%



\section{Conclusions}
The recovery of hierarchically sparse vectors from linear measurements of Kronecker type naturally appears in 
a plethora of practical applications. The HiHTP algorithm, an efficient algorithm for achieving such a
recovery, is guaranteed to work under a HiRIP condition. In this work, we have shown that a Kronecker product $\vec A \otimes \vec B$ has the $(s,\sigma)$-HiRIP property
as soon as its components exhibit the RIP. The analogous result holds for Kronnecker products with multiple factors and multi-level hierarchically sparse vectors. This is in contrast to the standard $s$-RIP, where each component needs to have the $s$-RIP. As a further application of our result, we described measurement schemes in which it can be efficiently checked to have the HiRIP, in sharp contrast to the general computational hardness of deciding whether a 
RIP constant is smaller than a given constant. These schemes, however, exhibit a suboptimal sample complexity.
 On a higher
level, the present work contributes to the program of identifying ways of achieving recovery of structured
vectors with as little
randomness as possible.

\section*{Acknowledgment}
AF acknowledges support from the DFG  (Grant KU 1446/18-1),
IR and JE by the DFG (EI 519/9-1), the Templeton Foundation and the ERC (TAQ), 
and GW by DFG SPP 1914 COSIP and EU H2020 5GPP project ONE5G.




\end{document}

%% file: mymath.tex
\newtheorem{theorem}{Theorem}
\newtheorem{corollary}[theorem]{Corollary}
\newtheorem{lemma}[theorem]{Lemma}

\newtheorem*{observation*}{Observation}
\theoremstyle{definition}

\newtheorem{definition}{Definition}

\DeclareMathOperator{\Tr}{Tr}
\renewcommand{\Re}{\operatorname{Re}}
\renewcommand{\Im}{\operatorname{Im}}

\DeclareMathOperator{\vecmap}{vec}

\DeclareMathOperator{\Id}{Id}
\DeclareMathOperator*{\argmin}{arg\,min}
\DeclareMathOperator{\supp}{supp}

\newcommand{\CC}{\mathbb{C}}
\newcommand{\RR}{\mathbb{R}}
\newcommand{\KK}{\mathbb{K}}

\newcommand{\NN}{\mathbb{N}}

\renewcommand{\vec}[1]{\mathbf{#1}}

\newcommand{\norm}[1]{\left\Vert #1 \right\Vert} 

\newcommand{\ad}{\ensuremath^\dagger}

\renewcommand{\A}{\vec A} 
\newcommand{\x}{\vec x} 
\newcommand{\y}{\vec y} 
\newcommand{\z}{\vec z} 
\newcommand{\ev}{\vec e} 

\newcommand{\proj}{\ensuremath\kern -0.12em\rfloor\kern -0.06em} 

\newlang{\bhtp}{HiHTP} 
\newlang{\HiHTP}{HiHTP} 
\newlang{\htp}{HTP} 
\newlang{\HTP}{HTP} 
\newlang{\GOMP}{GOMP}
\newlang{\HiLasso}{HiLasso}

